\documentclass[conference]{IEEEtran}
\usepackage{cite}
\usepackage{amsfonts,amssymb,amsthm,amsbsy,amsmath,paralist,bm,ifthen,color}
\usepackage{algorithm,algorithmic}
\usepackage{graphicx}
\usepackage{subfig}
\usepackage[font=footnotesize]{caption}
\usepackage{xcolor}
\usepackage{relsize}
\usepackage[numbers,sort&compress,square]{natbib}
\usepackage{tabularx}
\usepackage{booktabs}
\usepackage{multirow}
\usepackage{setspace}
\newtheorem{lemma}{Lemma}
\IEEEoverridecommandlockouts
\begin{document}
\captionsetup[figure]{labelfont={bf},name={Fig.},labelsep=period}
\title{RIS-aided D2D Communication Design for URLLC Packet Delivery}
\author{Jing Cheng$^\dag$, Chao Shen$^\dag$, Zheng Chen$^\ddag$, Nikolaos Pappas$^\S$\\
{\small $^\dag$State Key Laboratory of Rail Traffic Control and Safety, Beijing Jiaotong University, Beijing, China}\\
{\small $^\ddag$}Department of Electrical Engineering, Link\"oping University, 581 83 Link\"oping, Sweden\\
{\small $^\S$}Department of Science and Technology, Link\"oping University, Norrk\"oping, 602 21 Sweden\\
{\small Email: \{chengjing, chaoshen\}@bjtu.edu.cn, \{nikolaos.pappas, zheng.chen\}@liu.se}
\thanks{This work was supported in part by the Swedish Research Council (VR), ELLIIT, and CENIIT, in part supported by the the National Key R\&D Program of China (2020YFB1806604), the Fundamental Research Funds for the Central Universities +2020JBZD005, the NSFC, China (61871027, 62031008, U1834210) and the China Scholarship Council (CSC) Grant \#202007090174.}}
\maketitle

\begin{abstract}
  In this paper, we consider a smart factory scenario where a set of actuators receive critical control signals from an access point (AP) with reliability and low latency requirements. We investigate jointly active beamforming at the AP and passive phase shifting at the reconfigurable intelligent surface (RIS) for successfully delivering the control signals from the AP to the actuators within a required time duration. The transmission follows a two-stage design. In the first stage, each actuator can both receive the direct signal from AP and the reflected signal from the RIS. In the second stage, the actuators with successful reception in the first stage, relay the message through the D2D network to the actuators with failed receptions. We formulate a non-convex optimization problem where we first obtain an equivalent but more tractable form by addressing the problem with discrete indicator functions. Then, Frobenius inner product based equality is applied for decoupling the optimization variables. Further, we adopt a penalty-based approach to resolve the rank-one constraints. Finally, we deal with the $\ell_0$-norm by $\ell_1$-norm approximation and add an extra term $\ell_1-\ell_2$ for sparsity. Numerical results reveal that the proposed two-stage RIS-aided D2D communication protocol is effective for enabling reliable communication with latency requirements.
\end{abstract}
\section{Introduction}
Ultra-reliable and low-latency communication (URLLC) is one of the pillar technologies in $5$G and $6$G mobile communication networks. It is an enabler for a wide range of mission-critical applications including factory automation, smart grid, tactile internet, intelligent transportation system, telesurgery, and virtual/augmented reality (VR/AR) \cite{BennisPotI2018}. These applications pose great challenges for end-to-end packet delivery. More specifically, a URLLC packet has to be delivered to the destination successfully within sub-millisecond latency and extremely high reliability. This calls for new technologies and transmission mechanism design to satisfy such strict quality-of-service (QoS) requirements.

One promising technology is the reconfigurable intelligent surface (RIS) which has attracted interest recently. RIS, a metasurface composed of massive low-cost and passive reflecting elements, is able to reshape and reconfigure the wireless environment by intelligently adjusting the amplitudes and/or phase shifts of reflecting elements \cite{Tatino2020,LiaskosICM2018,Tatino2021}. As a result, integrating RIS into a communication system can improve the desired signal at the destination and facilitate reliable end-to-end communication. Moreover, as compared to active relays, RIS can work in a full-duplex mode without causing self-interference and thermal noise and its energy consumption is low owing to its passive nature \cite{WuIToWC2019}. Hence, the deployment of cost-effective RIS is important for enabling URLLC packet delivery. The authors in \cite{Ghanem2021} investigated the resource allocation design problem of a multicell RIS-aided OFDMA system for weighted sum throughput maximization subject to all URLLC users' QoS requirements.

Another potential technology is device-to-device (D2D) communication. In most mission-critical scenarios, devices (e.g. sensors, actuators, machines) are in close proximity to each other. Thus, the channel between devices is significantly more reliable than that between the access point (AP) and the devices \cite{XiaoIWC2016}. By exploiting the benefits of a D2D network, successful URLLC packet delivery can be enabled. In \cite{LiuITWC2018}, a D2D-based two-phase transmission protocol was proposed for industrial automation scenarios to achieve URLLC.

Recently, works considering the combination of both RIS and D2D technologies for system performance optimization have appeared. The authors in \cite{ChenIToWC2021} found that by integrating RIS into a D2D network as well as optimizing the transmit power and discrete phase shifts of RIS elements, the interference of D2D networks can be significantly decreased. In \cite{JiaIWCL2021}, authors investigated a problem of joint power control of D2D users and passive beamforming for an RIS-aided D2D communication network to maximize energy efficiency. Unfortunately, the existing RIS-aided D2D communication design cannot be applied to URLLC packet transmission directly since these studies were based on Shannon's capacity with assumptions of infinite blocklength and zero error probability. Thus, the delay performance will be underestimated and the reliability performance will be overestimated \cite{XuIToWC2020}. This necessitates the URLLC packet transmission design in an RIS-assisted D2D network under the finite blocklength regime.

In this paper, we consider an RIS-aided D2D network to achieve downlink multiusers' URLLC packet delivery in smart factory scenarios. The system involves two stages. In the first stage, each actuator's received signal consists of two parts -- a direct signal from the AP and a reflected signal from the RIS. In the second stage, actuators with successful reception in the first stage help relay critical messages to the ones who failed to receive via the D2D network and they use maximum-ratio combining (MRC) to jointly decode the message received in two stages. The main contributions of this paper are summarized as follows.
\begin{itemize}
  \item To the best of our knowledge, this paper is the first to consider the two-stage RIS-aided D2D communication protocol for URLLC packet delivery. To enable all actuators' reliable communication, a joint active and passive beamforming optimization problem is formulated.
  \item Instead of the common alternating optimization method, we first transform the original problem into an equivalent but more tractable form. Then we leverage the Frobenius inner product based equality for product decomposition and penalty-based method for tackling rank-one constraints. Eventually, we address the  $\ell_0$-norm by $\ell_1$-norm approximation plus $\ell_1-\ell_2$ for sparsity.
  \item A low-complexity penalized successive convex approximation (SCA) based iterative algorithm is proposed to obtain suboptimal solutions. Simulation results demonstrate the superiority of the proposed two-stage RIS-aided D2D communication protocol.
\end{itemize}
\section{System Model}
We consider a downlink communication system in a smart factory with one AP equipped with $N_t$ antennas and $K$ single-antenna actuators indexed by $k=\{1,\cdots,K\}$. All the actuators must successfully receive critical control messages in terms of short packets within a time window of $\tau$ seconds. As shown in Fig. \ref{systemModel}, there are $M$ reflecting elements on the RIS. The phase shift matrix of RIS is denoted by $\pmb{\Phi}=\mathrm{diag}(\phi_1,\cdots,\phi_M)$, where $\phi_m=e^{j\theta_m}$ and $\theta_m\in[0,2\pi]$ is the phase shift of the $m$-th reflecting element. Assume that the phase shifts of all elements are calculated by the AP and then fed back to the RIS controller through the dedicated RIS control link.
\begin{figure}[!htbp]
	\centering
	\includegraphics[width=.8\linewidth]{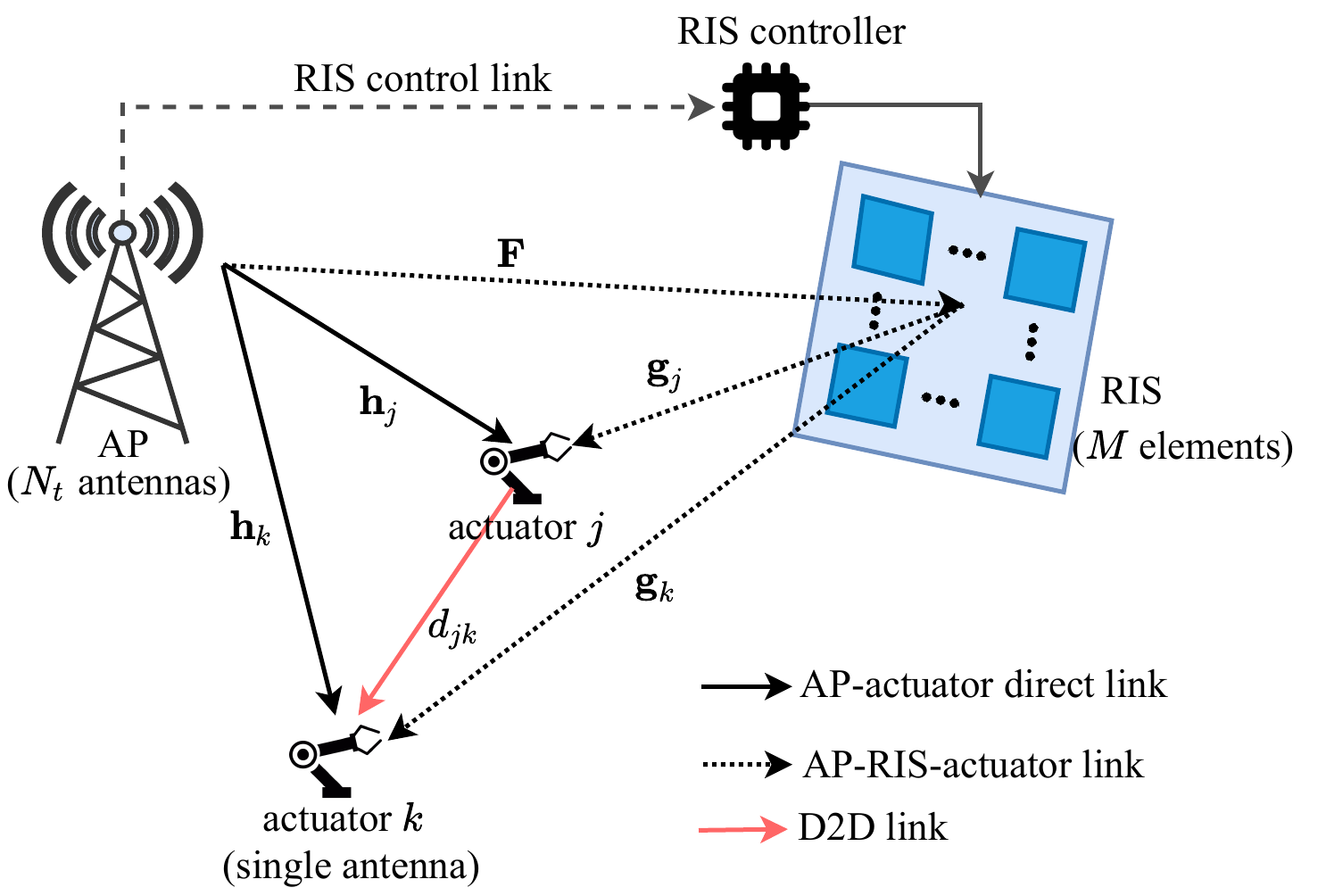}
	\caption{An illustrative example of the considered setup.}
	\label{systemModel}
\end{figure}

\vspace{-2mm}
Two stages are involved in the RIS-assisted D2D cooperative communication system. We assume that the two stages have the same transmission duration, i.e., $\tau_1=\tau_2=0.5\tau$. In the first stage with duration of $\tau_1$ seconds, each actuator can receive the direct signal from AP and the reflected signal from the RIS. In the second stage with duration of $\tau_2$ seconds, the actuators who have successfully decoded the signal in the first stage can act as relays to forward the signal to the actuators with failed recepetion via the D2D network, and they jointly decode the message received in two stages by MRC. The two-stage RIS-aided D2D communication protocol is described in detail as below.
\subsection{RIS-aided first-stage transmission}
In the first stage with duration of $\tau_1$ seconds, the transmitted signal at the AP is given by
\vspace{-1mm}
\begin{equation}
  \mathbf{x}^{(\text{I})}=\mathbf{w} s,\label{Tx_stageI}
\end{equation}
where $\mathbf{w}\in\mathbb{C}^{N_t\times 1}$ is the beamforming vector for broadcasting at the AP, and $s$ denotes the combined symbol intended for all actuators. Without loss of generality, we assume $\mathbb{E}\{|s|^2\}=1$. The received signal at the actuator $k$ can be expressed as
\vspace{-1mm}
\begin{equation}
  y_k^{(\text{I})}=\left(\mathbf{h}_k^H+\mathbf{g}_k^H\pmb{\Phi}\mathbf{F}\right)\mathbf{x}^{(\text{I})}
  +z_k^{(\text{I})},\label{Rx_stageI}
\end{equation}
where $\mathbf{h}_k\in\mathbb{C}^{N_t\times1},\mathbf{F}\in\mathbb{C}^{M\times N_t},\mathbf{g}_k\in\mathbb{C}^{M\times1}$ denote the channels from the AP to the actuator $k$, from the AP to the RIS, and from the RIS to the actuator $k$, respectively, $z_k^{(\text{I})}\sim\mathcal{CN}(0,\sigma_k^2)$ is the additive white Gaussian noise (AWGN).

By substituting \eqref{Tx_stageI} into \eqref{Rx_stageI}, we can obtain the signal-to-noise ratio (SNR) of the $k$-th actuator in the first stage as
\vspace{-2mm}
\begin{equation}
  \gamma_k^{(\text{I})}=\begin{vmatrix}\left(\bar{\mathbf{h}}_k^H+\bar{\mathbf{g}}_k^H\pmb{\Phi}\mathbf{F}\right)
  \mathbf{w}\end{vmatrix}^2,
\end{equation}
where $\bar{\mathbf{h}}_k=\frac{\mathbf{h}_k}{\sigma_k}$ and $\bar{\mathbf{g}}_k=\frac{\mathbf{g}_k}{\sigma_k}$.

Based on the normal approximation of the achievable rate for short packet communication under quasi-static AWGN channel conditions \cite{PolyanskiyITIT2010}, we can characterize the maximum achievable rate of actuator $k$ over $L_1=\tau_1B$ channel uses in the first stage by
\vspace{-2mm}
\begin{equation}
  \mathcal{C}_k^{(\text{I})}=\frac{D}{L_1}=\log_2\begin{pmatrix}1+\gamma_k^{(\text{I})}\end{pmatrix}
  -\sqrt{\frac{V_k^{(\text{I})}}{L_1}}
  Q^{-1}\begin{pmatrix}\varepsilon_k^{(\text{I})}\end{pmatrix},\label{FBL_stageI}
\end{equation}
where $B$ is the bandwidth, $D$ is the number of data bits, $V_k^{(\text{I})}=(\log_2 e)^2\begin{pmatrix}1-\begin{pmatrix}1+\gamma_k^{(\text{I})}\end{pmatrix}^{-2}\end{pmatrix}$ is the channel dispersion, $\varepsilon_k^{(\text{I})}$ is the decoding error probability, $Q^{-1}(\cdot)$ is the inverse of Q-function $Q(x)=\frac{1}{\sqrt{2 \pi}}\int_{x}^{\infty} e^{-t^{2}/2}dt$. This asymptotic expression characterizes the relationship between the achievable rate, transmission time (channel uses) and decoding error probability. From \eqref{FBL_stageI}, the decoding error probability can be written as
\vspace{-2mm}
\begin{equation}
  \varepsilon_k^{(\text{I})}\!=\!Q\!\left(\!\frac{\sqrt{L_1}\log_2\left(\!1
  \!+\!\gamma_k^{(\text{I})}\!\right)\!-\!D/\sqrt{L_1}}
  {\sqrt{V_k^{(\text{I})}}}\right)\!\triangleq\! Q\left(\!f\left(\gamma_k^{(\text{I})}\!\right)\right).\label{epsilon1}
\end{equation}

Let $\varepsilon_{\text{th}}$ denote the maximum allowed packet error probability (PEP). Note that $f\begin{pmatrix}\gamma_k^{(\text{I})}\end{pmatrix}$ is monotonically increases with $\gamma_k^{(\text{I})}$ \cite{MalakIToWC2019} and Q-function is a monotonically decreasing function. Thus, we can obtain that $\varepsilon_k^{(\text{I})}$ decreases with $\gamma_k^{(\text{I})}$. Therefore, the maximum PEP $\varepsilon_{\text{th}}$ corresponds to the minimum SNR $\gamma_{\text{th}}^{(\text{I})}$ required for successful decoding in the first stage.

In order to indicate which actuator successfully decodes the message in the first stage, we define an indicator as follows
\vspace{-2mm}
\begin{equation}
    a_k^{(\text{I})}=\left\{\begin{array}{ll}
    1,& \mathrm{if}~ \gamma_k^{(\text{I})}\geq \gamma_{\text{th}}^{(\text{I})},\\
    0,& \mathrm{otherwise}.
    \end{array}\right.
\end{equation}
\subsection{Second-stage D2D transmission}
In the second stage with duration of $\tau_2$ seconds, the actuators with successful decoding in the first stage help to relay the signal to others via the D2D network. Here, we assume that the actuators acting as relays simply forward the signal at their full power, i.e., power control is not performed since the global CSI of the D2D network is not available. The transmit signal of the actuator $k$ in the second stage can be given by
\vspace{-1mm}
\begin{equation}
  x_k^{(\text{II})}=a_k^{(\text{I})}\sqrt{P}s,
\end{equation}
where $P$ is the common transmit power for all actuators. Then, in the second stage, if the actuator $k$ cannot decode the message successfully in the first stage, its received signal can be expressed as
\vspace{-1mm}
\begin{equation}
  y_k^{(\text{II})}=\sum\limits_{j\neq k}d_{jk}x_j^{(\text{II})}+z_k^{(\text{II})},
\end{equation}
where $d_{jk}\in\mathbb{C}$ is the channel from the actuator $j$ to actuator $k$, $z_k^{(\text{II})}\sim\mathcal{CN}(0,\sigma_k^2)$ is the AWGN. At the end of this stage, the actuator $k$ jointly decodes the message based on the signals received in both two stages by
using MRC. Accordingly, the received SNR at actuator $k$ can be written as
\vspace{-2mm}
\begin{equation}
  \gamma_k^{(\text{II})}=\begin{vmatrix}\left(\bar{\mathbf{h}}_k^H+\bar{\mathbf{g}}_k^H\pmb{\Phi}\mathbf{F}\right)
  \mathbf{w}\end{vmatrix}^2+P
  \begin{vmatrix} \sum\limits_{j\neq k}a_j^{(\text{I})}\bar{d}_{jk}  \end{vmatrix}^2,
\end{equation}
where $\bar{d}_{jk}=\frac{d_{jk}}{\sigma_k}$. Then, the maximum achievable rate of actuator $k$ over $L_2=\tau_2B$ channel uses in the second stage can be characterized by
\vspace{-2mm}
\begin{equation}
  \mathcal{C}_k^{(\text{II})}=\frac{D}{L_2}=\log_2\begin{pmatrix}1+\gamma_k^{(\text{II})}\end{pmatrix}
  -\sqrt{\frac{V_k^{(\text{II})}}{L_2}}
  Q^{-1}\begin{pmatrix}\varepsilon_k^{(\text{II})}\end{pmatrix},\label{FBL_stageII}
\end{equation}
where $\varepsilon_k^{(\text{II})}$ is the decoding error probability of actuator $k$ in the second stage. Likewise, based on equation \eqref{FBL_stageII}$, \varepsilon_k^{(\text{II})}$ can be represented as
\vspace{-2mm}
\begin{equation}
  \varepsilon_k^{(\text{II})}\!=\!Q\!\left(\!\frac{\sqrt{L_2}\log_2\left(\!1
  \!+\!\gamma_k^{(\text{II})}\!\right)\!-\!D/\sqrt{L_2}}
  {\sqrt{V_k^{(\text{II})}}}\right)\!\triangleq\! Q\left(\!f\left(\gamma_k^{(\text{II})}\!\right)\right).\label{epsilon2}\\
\end{equation}
Similar to our observation from \eqref{epsilon1}, it is easy to obtain that $\varepsilon_k^{(\text{II})}$ decreases with $\gamma_k^{(\text{II})}$. Hence, the maximum PEP threshold $\varepsilon_{\text{th}}$ corresponds to the minimum SNR $\gamma_{\text{th}}^{(\text{II})}$ requirement for successful reception in the second stage. Let $a_k^{(\text{II})}$ indicate whether actuator $k$ failed in the first stage can complete the reception successfully in the second stage. More specifically, it is described as
\vspace{-2mm}
\begin{equation}
    a_k^{(\text{II})}=\left\{\begin{array}{ll}
    1,& \mathrm{if}~ a_k^{(\text{I})}=0,\gamma_k^{(\text{II})}\geq \gamma_{\text{th}}^{(\text{II})},\\
    0,& \mathrm{if}~ a_k^{(\text{I})}=0,\gamma_k^{(\text{II})}< \gamma_{\text{th}}^{(\text{II})}.
    \end{array}\right.
\end{equation}
\section{Problem Formulation}
In this paper, we aim to jointly design the active and passive beamforming at the AP and RIS respectively to maximize the number of actuators that can receive the critical control signals successfully under the two-stage RIS-aided D2D communication protocol. Mathematically, the joint optimization problem can be formulated as
\vspace{-2mm}
\begin{subequations}
\begin{align}
\mathrm{P1}:\max_{\mathbf{w},\pmb{\Phi}} \quad&\sum_{k=1}^K \begin{pmatrix}a_k^{(\text{I})}+a_k^{(\text{II})}\end{pmatrix}\\
\mathrm{s.t.}\quad~ &\|\mathbf{w}\|^2\leq P_{\text{max}},\label{powerConstraint}\\
&|\phi_m|=1,\forall m.\label{unit-modulus}
\end{align}
\end{subequations}
Constraint \eqref{powerConstraint} is the transmit power budget at the AP. The constraint \eqref{unit-modulus} refers to the unit-modulus constraint for each RIS element.

Our ultimate goal is to guarantee that all actuators can successfully receive the mission-critical control signals from AP after the two-stage transmission, i.e., $\sum_{k=1}^K \begin{pmatrix}a_k^{(\text{I})}+a_k^{(\text{II})}\end{pmatrix}=K$. Note that the larger the number of actuators with successful reception in the first stage is, the higher the probability of the remaining actuators having successful reception in the second stage with the aid of D2D network will be. This is mainly because of the high SNR of the direct D2D link between two actuators in proximity. Moreover, actuators which act as relays in the second stage are determined by the optimization results in the first stage. Thus, the original problem can be transformed into an alternative problem that to maximize the number of successful decoding actuators in the first stage. As a result, fewer remaining actuators are required to reach the SNR threshold $\gamma_{\text{th}}^{(\text{II})}$ for reliable reception in the second stage. Therefore, problem $\mathrm{P1}$ can be reformulated as
\vspace{-2mm}
\begin{subequations}
\begin{align}
\!\!\!\mathrm{P2}:\max_{\mathbf{w},\pmb{\Phi}} \quad&\sum_{k=1}^K a_k^{(\text{I})}\\
\mathrm{s.t.}\quad~ &\eqref{powerConstraint},\eqref{unit-modulus}.
\end{align}
\end{subequations}
\section{Suboptimal Joint Active and Passive Beamforming Design}
In this section, the suboptimal joint active and passive beamforming design for problem $\mathrm{P2}$ is given.
\subsection{Problem Transformation}
Note that in problem $\mathrm{P2}$, $a_k^{(\text{I})}$ is a discrete function of $\gamma_k^{(\text{I})}$. The discrete nature makes it difficult to apply the standard optimization techniques. In view of this, we introduce a set of auxiliary variables $\mathbf{q}=[q_1,\cdots,q_k]^T$ to transform problem $\mathrm{P2}$ into an equivalent but more tractable form, which can be written as
\vspace{-2mm}
\begin{subequations}
\begin{align}
\mathrm{P3}:\min_{\mathbf{q,w},\pmb{\Phi}} \quad&\|\mathbf{q}\|_0\\
\mathrm{s.t.}~~~~ &\gamma_k^{(\text{I})}+q_k\geq \gamma_{\text{th}}^{(\text{I})},\forall k,\label{P3-C1}\\
&q_k\geq 0,\forall k,\label{P3-C2}\\
&\eqref{powerConstraint},\eqref{unit-modulus}.
\end{align}
\end{subequations}
Here, nonzero element $q_k$ means the actuator $k$ cannot decode the critical signal successfully. Consequently, maximizing the number of successful actuators in the first stage is equivalent to minimizing the number of nonzero elements in vector $\mathbf{q}$.

Now, our focus turns to solve problem $\mathrm{P3}$. Its difficulty mainly includes $\ell_0$-norm in the objective function, the non-convex constraint \eqref{P3-C1} due to the coupling of $\mathbf{w}$ and $\pmb{\Phi}$ in the form of a product, and unit-modulus constraint \eqref{unit-modulus}.

First, to address the coupling problem, let $\mathbf{R}_k=\left[\mathbf{F}^H\mathrm{diag}(\bar{\mathbf{g}}_k)\quad \bar{\mathbf{h}}_k\right]\in\mathbb{C}^{N_t\times (M+1)}$, $\tilde{\mathbf{v}}=[\phi_1,\cdots,\phi_M]^H\in\mathbb{C}^{M\times 1}$, and $\mathbf{v}=[\tilde{\mathbf{v}}^T \quad1]^T$, then we can equivalently transform $\gamma_k^{(\text{I})}$ as
\begin{equation}
  \gamma_k^{(\text{I})}=\mathbf{v}^H\mathbf{R}_k^H\mathbf{w}\mathbf{w}^H\mathbf{R}_k\mathbf{v}
  =\mathrm{Tr}(\mathbf{W}\mathbf{R}_k\mathbf{V}\mathbf{R}_k^H),\label{gammakI}
\end{equation}
where $\mathbf{W}=\mathbf{w}\mathbf{w}^H,\mathbf{V}=\mathbf{v}\mathbf{v}^H$. Here, $\mathbf{W}$ and $\mathbf{V}$ have to satisfy the positive semidefinite constraints $\mathbf{W}\succeq \mathbf{0},\mathbf{V}\succeq \mathbf{0}$, and the rank-one constraints $\mathrm{rank}(\mathbf{W})=1,\mathrm{rank}(\mathbf{V})=1$.

It is necessary to remark that $\mathbf{W}$ and $\mathbf{V}$ remain coupled together in the form of a product in \eqref{gammakI}. However, we can leverage the following equality to decompose the product, which is given by
\begin{equation}
  \langle\mathbf{A,B}\rangle_F=\frac{1}{2}\left(\|\mathbf{A}+\mathbf{B}\|_{F}^{2}-\|\mathbf{A}\|_{F}^{2}
  -\|\mathbf{B}\|_{F}^{2}\right),
\end{equation}
where $\langle\mathbf{A,B}\rangle_F=\mathrm{Tr}(A^HB)$ is the Frobenius inner product. Here, $\mathbf{W}$ and $\mathbf{R}_k\mathbf{V}\mathbf{R}_k^H$ are both $N_t$-by-$N_t$ matrices. By utilizing the Hermitian property of $\mathbf{W}$, i.e., $\mathbf{W}=\mathbf{W}^H$, we can exercise the above equality for product decomposition. Therefore, $\gamma_k^{(\text{I})}$ can be equivalently written as
\begin{equation}\label{Frobenreform}
  \gamma_k^{(\text{I})}\!=\!\frac{1}{2}\!\left(\|\mathbf{W}\!+\!\mathbf{R}_k\mathbf{V}\mathbf{R}_k^H\|_{F}^{2}\!-\!\|\mathbf{W}\|_{F}^{2}
  \!-\!\|\mathbf{R}_k\mathbf{V}\mathbf{R}_k^H\|_{F}^{2}\right)\!.
\end{equation}

Since $\gamma_k^{(\text{I})}$ is not joint concave with respect to $\mathbf{W},\mathbf{V}$, based on the first-order condition, we can obtain the lower bound of the convex function $\|\mathbf{W}+\mathbf{R}_k\mathbf{V}\mathbf{R}_k^H\|_{F}^{2}$, which is given by
\begin{align}
  &f_k\left(\mathbf{W},\mathbf{V}\right)=\|\mathbf{W}^{(i)}
  +\mathbf{R}_k\mathbf{V}^{(i)}\mathbf{R}_k^H\|_{F}^{2}\notag\\
  &\!+\!2\mathrm{Tr}\!\begin{pmatrix}\mathrm{Re}\!\begin{Bmatrix}\!\begin{pmatrix}\mathbf{W}^{(i)}\!
  +\!\mathbf{R}_k\!\mathbf{V}^{(i)}\mathbf{R}_k^H\!\end{pmatrix}\!\begin{pmatrix}\mathbf{W}\!
  -\!\mathbf{W}^{(i)}\!\end{pmatrix}\!\end{Bmatrix}\!\end{pmatrix}\notag\\
  &\!+\!\!2\mathrm{Tr}\!\begin{pmatrix}\mathrm{Re}\!
  \begin{Bmatrix}\!\begin{pmatrix}\mathbf{R}_k^H\!\begin{pmatrix}\mathbf{W}^{(i)}\!
  +\!\mathbf{R}_k\!\mathbf{V}^{(i)}
  \mathbf{R}_k^H\end{pmatrix}\!\mathbf{R}_k\!\end{pmatrix}\!\!
  \begin{pmatrix}\mathbf{V}\!\!
  -\!\!\mathbf{V}^{(i)}\!\end{pmatrix}\!\end{Bmatrix}\!\end{pmatrix},
\end{align}
where $\mathbf{W}^{(i)},\mathbf{V}^{(i)}$ are feasible points in the $i$-th iteration. Thus, $\gamma_k^{(\text{I})}$ can be approximated as
\begin{equation}
  \bar{\gamma}_k^{(\text{I})}\!=\!\frac{1}{2}\!\left(f_k\left(\mathbf{W},\mathbf{V}\right)\!-\!\|\mathbf{W}\|_{F}^{2}
  \!-\!\|\mathbf{R}_k\mathbf{V}\mathbf{R}_k^H\|_{F}^{2}\right)\!.
\end{equation}
\begin{lemma}\label{lemma}
 For any $\mathbf{A}\in\mathbb{H}^{n}$, the constraint that $\mathbf{A}$ is rank-one is equivalent to
  \begin{equation}
    \|\mathbf{A}\|_{*}- \|\mathbf{A}\|_2\leq0.\label{Nuclear_spectal}
  \end{equation}
\end{lemma}
\begin{proof}
  For any $\mathbf{A}\in\mathbb{H}^{n}$, the inequality holds $\|\mathbf{A}\|_{*}=\sum_{i}\sigma_i\geq\|\mathbf{A}\|_2=\max_{i}\{\sigma_i\}$, where $\sigma_i$ is the $i$-th singular value of matrix $\mathbf{A}$. And the equality holds if and only if $\mathbf{A}$ is rank-one. Thus, the implicit constraint of Hermitian matrix $\mathbf{A}$, i.e., $\|\mathbf{A}\|_{*}- \|\mathbf{A}\|_2\geq0$ functions simultaneously with the constraint \eqref{Nuclear_spectal}, which requires that $\|\mathbf{A}\|_{*}- \|\mathbf{A}\|_2=0$, namely $\mathbf{A}$ is a rank-one matrix. This completes the proof.
\end{proof}

According to Lemma \ref{lemma}, rank-one constraints of $\mathbf{W,V}$ can be equivalently reformulated as
\begin{equation}
  \|\mathbf{W}\|_{*}-\|\mathbf{W}\|_2\leq 0,~\|\mathbf{V}\|_{*}-\|\mathbf{V}\|_2\leq 0.\label{nuclearSpectral}
\end{equation}
Then we adopt the penalty-based method by moving the constraint \eqref{nuclearSpectral} into the objective function, thereby resulting in the following optimization problem
\vspace{-2mm}
\begin{subequations}
\begin{align}
\!\!\!\!\!\mathrm{P4}:\min_{\mathbf{q,W,V}} ~~&\|\mathbf{q}\|_0\!\!+\!\!\alpha(\|\mathbf{W}\|_{*}\!\!-\!\!\|\mathbf{W}\|_2)\!\!
+\!\!\beta(\|\mathbf{V}\|_{*}\!-\!\|\mathbf{V}\|_2)\\
\mathrm{s.t.}\quad~ &\bar{\gamma}_k^{(\text{I})}+q_k\geq \gamma_{\text{th}}^{(\text{I})},\forall k,\label{P4-C1}\\
&q_k\geq 0,\forall k,\\
&\mathrm{Tr}(\mathbf{W})\leq P_{\text{max}},\\
&\mathbf{W}\succeq \mathbf{0},\mathbf{V}\succeq \mathbf{0},\label{PSD}\\
&[\mathbf{V}]_{mm}=1,m=1,\cdots,M+1,\label{diagV}
\end{align}
\end{subequations}
where $\alpha,\beta>0$ are penalty factors penalizing the violation of the constraint \eqref{nuclearSpectral}. Since penalty terms are in a form of difference-of-convex (DC) function, we apply the first-order Taylor series approximation to $\|\mathbf{W}\|_2$ and $\|\mathbf{V}\|_2$, which are given by
\begin{equation}
  \!f_w\!\!\triangleq\!\!\|\mathbf{W}^{(i)}\|_2\!\!+\!\!\mathrm{Tr}\!\begin{pmatrix}\!\mathrm{Re}\!\begin{Bmatrix}\!
  \pmb{\lambda}_{\max}(\mathbf{W}^{(i)})\pmb{\lambda}_{\max}^H(\mathbf{W}^{(i)})(\mathbf{W}\!\!-
  \!\!\mathbf{W}^{(i)})
  \!\end{Bmatrix}\!\end{pmatrix}\!,
\end{equation}
\begin{equation}
  \!f_v\!\!\triangleq\!\!\|\mathbf{V}^{(i)}\|_2\!\!+\!\!\mathrm{Tr}\!\begin{pmatrix}\!\mathrm{Re}\!\begin{Bmatrix}\!
  \pmb{\lambda}_{\max}(\mathbf{V}^{(i)})\pmb{\lambda}_{\max}^H(\mathbf{V}^{(i)})(\mathbf{V}\!\!-
  \!\!\mathbf{V}^{(i)})\!
  \end{Bmatrix}\!\end{pmatrix}\!,
\end{equation}
where $\pmb{\lambda}_{\max}(\mathbf{W})$ denotes the eigenvector corresponding to the largest eigenvalue of matrix $\mathbf{W}$, and $\mathbf{W}^{(i)},\mathbf{V}^{(i)}$ are feasible points in the $i$-th iteration. Hence, the penalty terms related to $\mathbf{W,V}$ can be represented as $\Gamma(\mathbf{W,V})=\alpha(\|\mathbf{W}\|_{*}\!-\!f_w)\!
+\!\beta(\|\mathbf{V}\|_{*}\!-\!f_v)$.

Finally, let us deal with the $\ell_0$-norm in the objective function. We first utilize $\|\mathbf{q}\|_1$ to approximate $\|\mathbf{q}\|_0$. Note that this approximation may sometimes be loose. In view of this, we add a penalty term $\eta(\|\mathbf{q}\|_1-\|\mathbf{q}\|_2)$ to make the number of zero elements in vector $\mathbf{q}$ as many as possible, where $\eta>0$ is a penalty parameter. In general, $\|\mathbf{q}\|_1\geq \|\mathbf{q}\|_2$ always holds true where the equality holds when vector $q$ has at most one nonzero element. The two terms $\|\mathbf{q}\|_1$ and $\|\mathbf{q}\|_1- \|\mathbf{q}\|_2$ function simultaneously. As a result,  the number of zero elements in vector $\mathbf{q}$ can be maximized while minimizing the value of nonzero elements (if exists) so that the SNR requirement in the second stage can be easier to reach. Note that $\|\mathbf{q}\|_1-\|\mathbf{q}\|_2$ is in a form of DC function, we exploit the first-order Taylor series approximation of $\|\mathbf{q}\|_2$. Thus, the penalty term can be approximated as
\vspace{-2mm}
\begin{equation}
  \Xi(\mathbf{q})\triangleq \eta\begin{pmatrix}\|\mathbf{q}\|_1
-\begin{pmatrix}\mathbf{q}^{(i)}\end{pmatrix}^T\mathbf{q}/\begin{Vmatrix}\mathbf{q}^{(i)}\end{Vmatrix}_2\end{pmatrix},
\end{equation}
where $\mathbf{q}^{(i)}$ is a feasible point in the $i$-th iteration. Therefore, we have the following penalized convex optimization problem
\begin{subequations}
\begin{align}
\mathrm{P5}:\min_{\mathbf{q,W,V}} \quad&\Delta\triangleq\|\mathbf{q}\|_1+\Gamma(\mathbf{W,V})+\Xi(\mathbf{q})\\
\mathrm{s.t.}\quad~ &\eqref{P4-C1}-\eqref{diagV},
\end{align}
\end{subequations}
which can be solved by standard convex tools like CVX.
\subsection{Algorithm Design}
According to the preceding analysis, we propose a penalized SCA based algorithm to solve problem $\mathrm{P2}$, which is summarized in Algorithm \ref{Alg1}. By iteratively solving a sequence of problem $\mathrm{P5}$ at feasible points, we can obtain a stationary-point solution of problem $\mathrm{P2}$ after convergence \cite{RazaviyaynSJO2012}.
\begin{algorithm}
\footnotesize
\caption{\small \textbf{:} Penalized SCA based Algorithm for Solving $\mathrm{P2}$}
\label{Alg1}
\begin{spacing}{1.2}
	\begin{algorithmic}[1]
		\STATE \textbf{Initialize}~iteration index $i=0$, feasible $\mathbf{q}^{(0)},\mathbf{W}^{(0)},\mathbf{V}^{(0)}$.
        \STATE Use bisection search method to obtain $\gamma_{\text{th}}^{(\text{I})}$.
        \STATE Calculate $\pmb{\lambda}_{\max}(\mathbf{W}^{(0)}),\pmb{\lambda}_{\max}(\mathbf{V}^{(0)})$.
		\REPEAT
        \STATE Set $i=i+1$.
		\STATE Obtain $\{\mathbf{q,W,V}\}$ by solving $\mathrm{P5}$.
        \STATE Update $\mathbf{q}^{(i)}=\mathbf{q},\mathbf{W}^{(i)}=\mathbf{W},\mathbf{V}^{(i)}=\mathbf{V}$.
        \STATE Update $\pmb{\lambda}_{\max}(\mathbf{W}^{(i)}),\pmb{\lambda}_{\max}(\mathbf{V}^{(i)})$.		
        \UNTIL $|\Delta^{(i)}-\Delta^{(i-1)}|\leq \epsilon$.
	\end{algorithmic}
\end{spacing}
\end{algorithm}
\section{Numerical Results}
In this section, we evaluate the performance of the proposed two-stage RIS-aided D2D communication protocol. The latency requirement is set as $\tau=1$ ms, each stage with $0.5$ ms. We consider the AP equipped with $N_t=4$ antennas is located at ($0$ m, $0$ m) transmitting critical control signals to $K=10$ actuators. All actuators are randomly and uniformly distributed on a circle centered at ($80$ m, $0$ m) with radius $20$ m. The RIS with $M=16$ reflecting elements is deployed at ($50$ m, $10$ m).

The channel $\mathbf{h}_k$ form AP to actuator $k$ is modeled as
\begin{equation}
  \mathbf{h}_k=\sqrt{L_0d_{a,k}^{-\alpha_{a,k}}}
  \sqrt{\frac{\beta_{a,k}}{1+\beta_{a,k}}}\mathbf{h}_k^{\mathrm{LoS}}+
  \sqrt{\frac{1}{1+\beta_{a,k}}}\mathbf{h}_k^{\mathrm{NLoS}},
\end{equation}
where $L_0\!=\!(\lambda_c/4\pi)^2$ is a constant with $\lambda_c$ being the wavelength of the carrier frequency (2.4 GHz), $d_{a,k}$ is the distance form AP to actuator $k$, $\alpha_{a,k}$ is the corresponding path loss exponent, $\beta_{a,k}$ is a Rician factor. $\mathbf{h}_k^{\mathrm{LoS}}$ and $\mathbf{h}_k^{\mathrm{NLoS}}$ represent the deterministic LoS and Rayleigh fading components, respectively. The channels from AP to RIS and from RIS to actuators are also generated in a similar manner. The path loss exponents for AP-actuator, AP-RIS, RIS-actuator channels are $\alpha_{a,k}\!=\!4,\alpha_{a,r}\!=\!\alpha_{r,k}\!=\!2$. The Rician factors for AP-actuator, AP-RIS, RIS-actuator channels are $\beta_{a,k}\!=\!0,\beta_{a,r}\!=\!\beta_{r,k}\!=\!4$. The large-scale fading of the D2D link is $35.3+37.6\log_{10}(d)$ in dB and the corresponding small-scale fading is modeled as Rician fading with Rician factor 4. Other parameters are listed in Table \ref{table1}.
\vspace{-2mm}
\begin{table}[htbp]\footnotesize{
\centering
\caption{Simulation Parameters}
\label{table1}
\begin{tabular}{lll}
  \toprule
  Symbol & Parameter & Value\\
  \midrule
  $B$ & bandwidth & 0.5 MHz\\
  $D$ & number of data bits & 100 $\sim$ 900\\
  $P_{\max}$ &maximum transmit power at the AP  & 43 dBm\\
  $P$ & transmit power of each actuator & 23 dBm\\
  $\sigma_k^2,\forall k$ & noise power & -80 dBm\\
  $\varepsilon_{\text{th}}$ &maximum PEP &$10^{-6}$\\
  $\alpha,\beta,\eta$ &penalty factors &$10^{2}$\\
  $\epsilon$ &convergence tolerance &$10^{-4}$\\
  \bottomrule
\end{tabular}}
\end{table}

In Fig. \ref{PRC_contrast}, we investigate the probability of reliable communication (PRC) under the following 4 transmission schemes. Here, PRC is defined as the ratio of number of experiments where all actuators can decode the critical control signal over total experiments.
\begin{enumerate}
  \item RIS (MRC): RIS is integrated for assisting the first stage transmission and MRC is adopted for decoding in the second stage.
  \item RIS (no MRC): This is similar to scheme 1 while it directly decodes the received message in the the second stage without exploiting that received in the first stage.
  \item No RIS (MRC): It only considers the two-stage D2D communication where RIS does not play a role in it. Meanwhile, MRC is adopted for decoding in stage II.
  \item No RIS (no MRC): It is similar to scheme 3 whereas it decodes the message directly in the second stage.
\end{enumerate}

From Fig. \ref{PRC_contrast}, we can easily find that when RIS is integrated into the communication system, all actuators can achieve reliable communication within a larger range of data bits (i.e. $700$ bits in transmission schemes 1 and 2, whereas without assistance of RIS, all actuators can only communicate reliably within the range of $500$ bits in transmission scheme $3$ and $400$ bits in scheme 4. This demonstrates that RIS plays an important role in promoting reliable communication and URLLC-related applications. Moreover, it is observed that MRC has a more significant effect on the no RIS-assisted transmission schemes 3 and 4 than that on the RIS-aided schemes 1 and 2. This is because RIS can greatly improve the received SNR of the desired signal at the actuator. More specifically, for problem $\mathrm{P5}$, RIS aims at maximizing the number of actuators with successful decoding in the first stage while minimizing the gap of failed actuators' received SNR and SNR threshold $\gamma_{\text{th}}^{(\text{I})}$ so that the failed actuators can more easily reach the target SNR threshold $\gamma_{\text{th}}^{(\text{II})}$ in the second stage.
\vspace{-3mm}
\begin{figure}[!htbp]\centering
	\includegraphics[width=.6\linewidth]{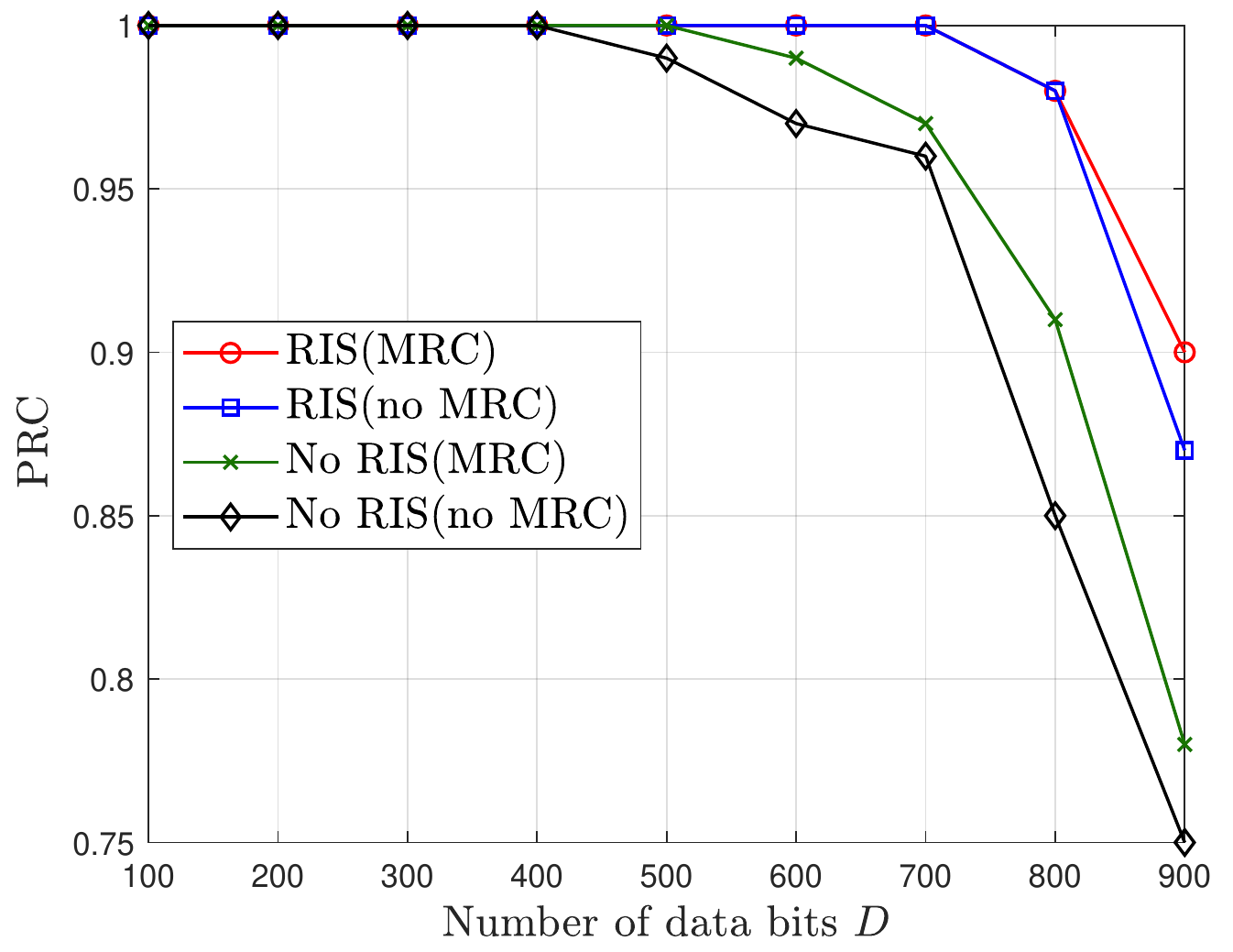}
	\caption{Probability of reliable communication versus number of data bits under different transmission schemes.}
	\label{PRC_contrast}
\end{figure}

\vspace{-2mm}
To further understand the two-stage RIS-assisted D2D communication protocol, we plot the average number of actuators with successful decoding in stage I and II under RIS and no RIS-assisted transmission schemes for different $D$ in Fig. \ref{Ptot_vs_reliability}. As can be seen from Fig. \ref{Ptot_vs_reliability} (a), when the number of data bits is small (i.e. not larger than $400$ bits), all actuators can just rely on the first-stage RIS-aided transmission to achieve reliable communication. As the number of data bits $D$ increases, second-stage D2D communication starts to function for message relaying towards the failed actuators in the first stage through strong D2D links. As a contrast, actuators have to depend on the D2D network for reliable communication from $D=300$ bits under the no RIS-assisted scheme, which is shown in Fig. \ref{Ptot_vs_reliability} (b). Thus, whether the RIS is introduced or not, when the number of data bits is large, the second-stage D2D-assisted transmission is indispensable. But with RIS, the burden of D2D network is lower and it will be easier to enable reliable communication, namely $\sum_{k=1}^K \begin{pmatrix}a_k^{(\text{I})}+a_k^{(\text{II})}\end{pmatrix}=K$.
\begin{figure}[!htbp]
\vspace{-2mm}
\centering
	\subfloat[with RIS]{\includegraphics[width=.6\linewidth]{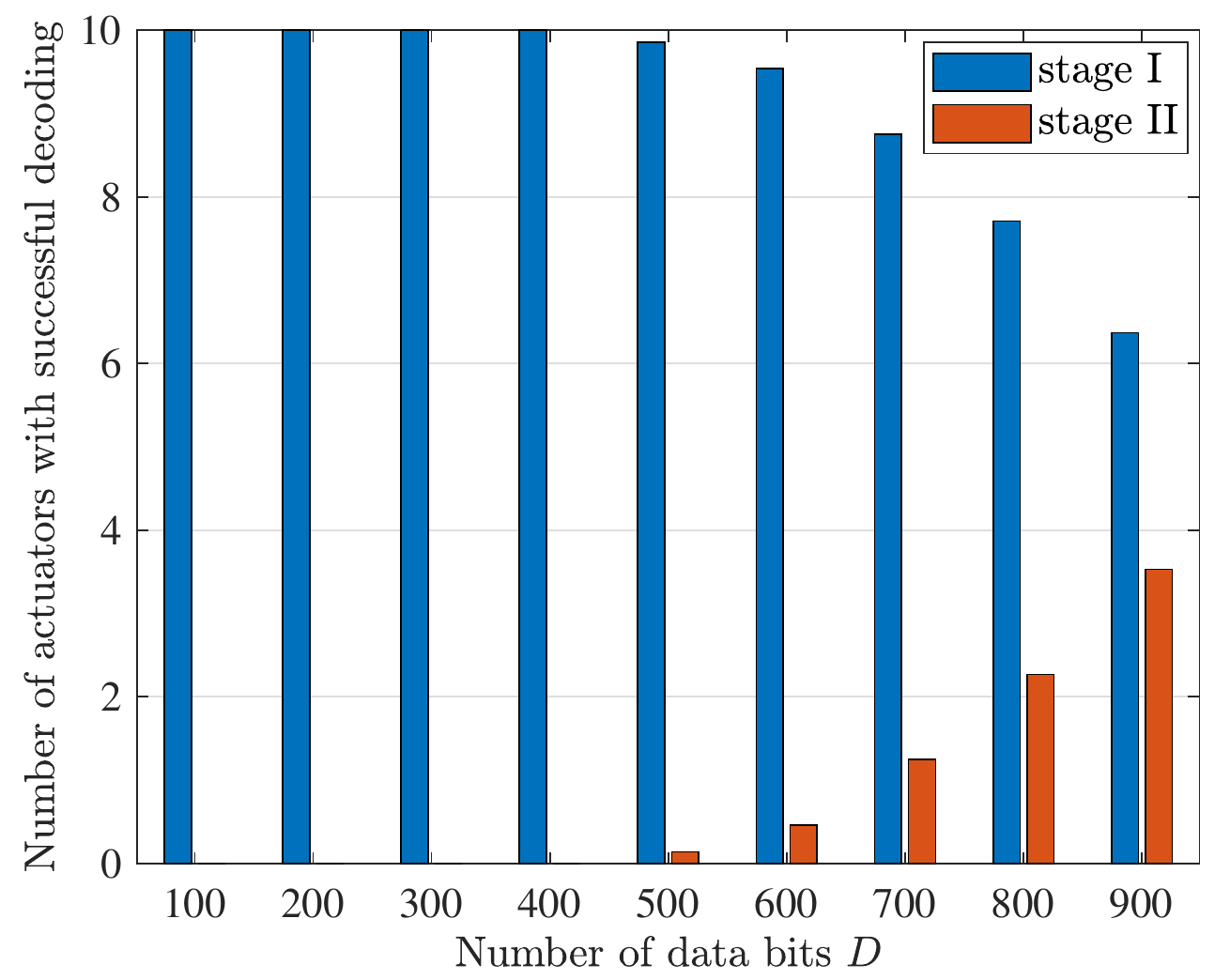}}\\[0.01mm]
	\subfloat[without RIS]{\includegraphics[width=.6\linewidth]{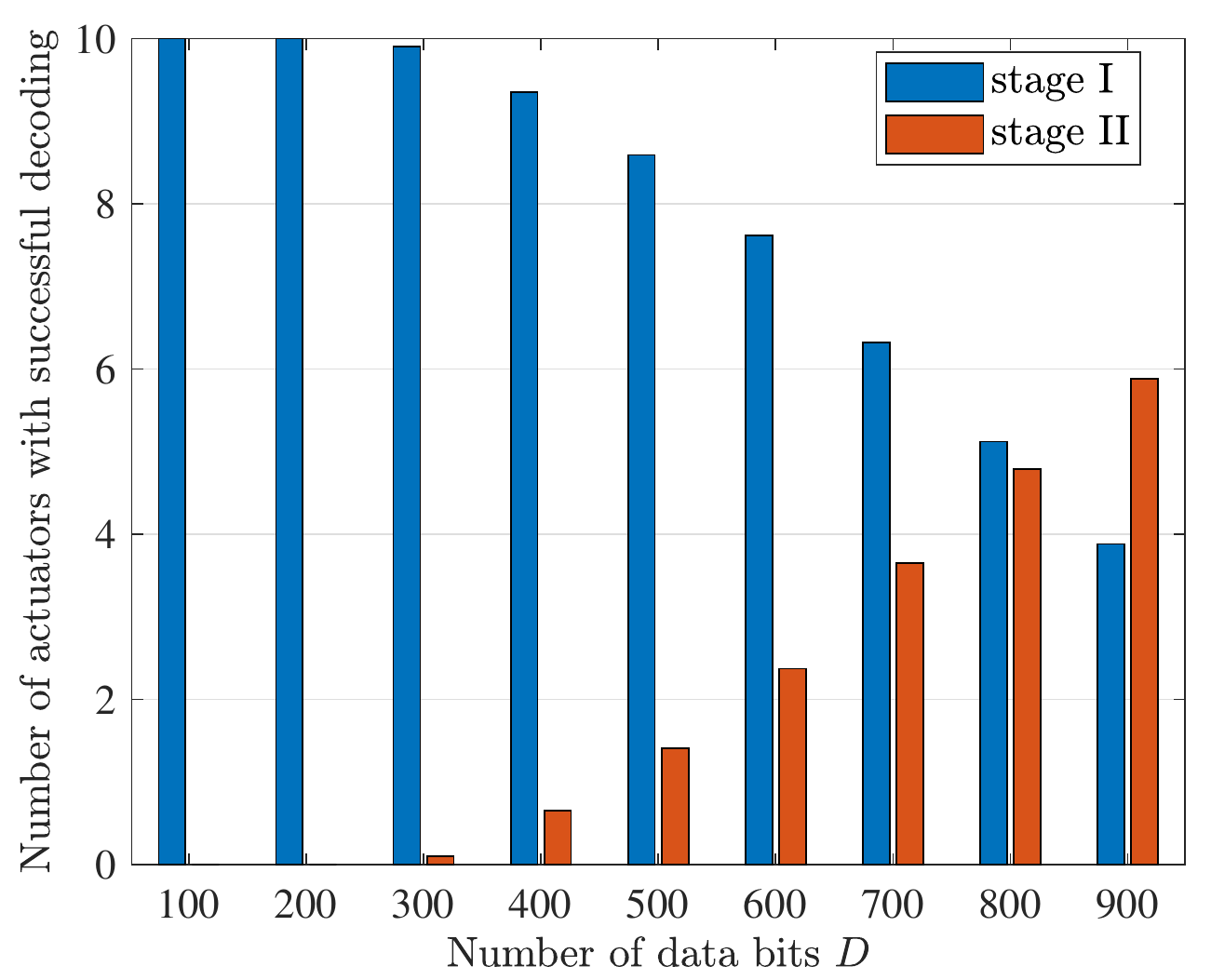}}
	\caption{Average number of actuators with successful decoding in each stage under RIS and no RIS-assisted transmission schemes for different $D$.}
	\label{Ptot_vs_reliability}
\end{figure}

Finally, in Fig. \ref{power_contrast}, we plot the average transmit power required at the AP under RIS and no RIS-assisted schemes for different number of data bits. As expected, for both schemes, the required transmit power of AP increases as the number of data bits increases, and when $D$ is relatively large, AP serves all actuators with full power for reliable communication. Notably, as compared to the no RIS-assisted scheme, the transmit power of AP can be considerably reduced by integrating the RIS into the communication system when the number of data bits is not larger than $600$. This shows another superiority of introducing RIS other than improving the quality of the received signal.
\begin{figure}[!htbp]
\setlength{\belowcaptionskip}{-0.6cm}
\centering
	\includegraphics[width=.6\linewidth]{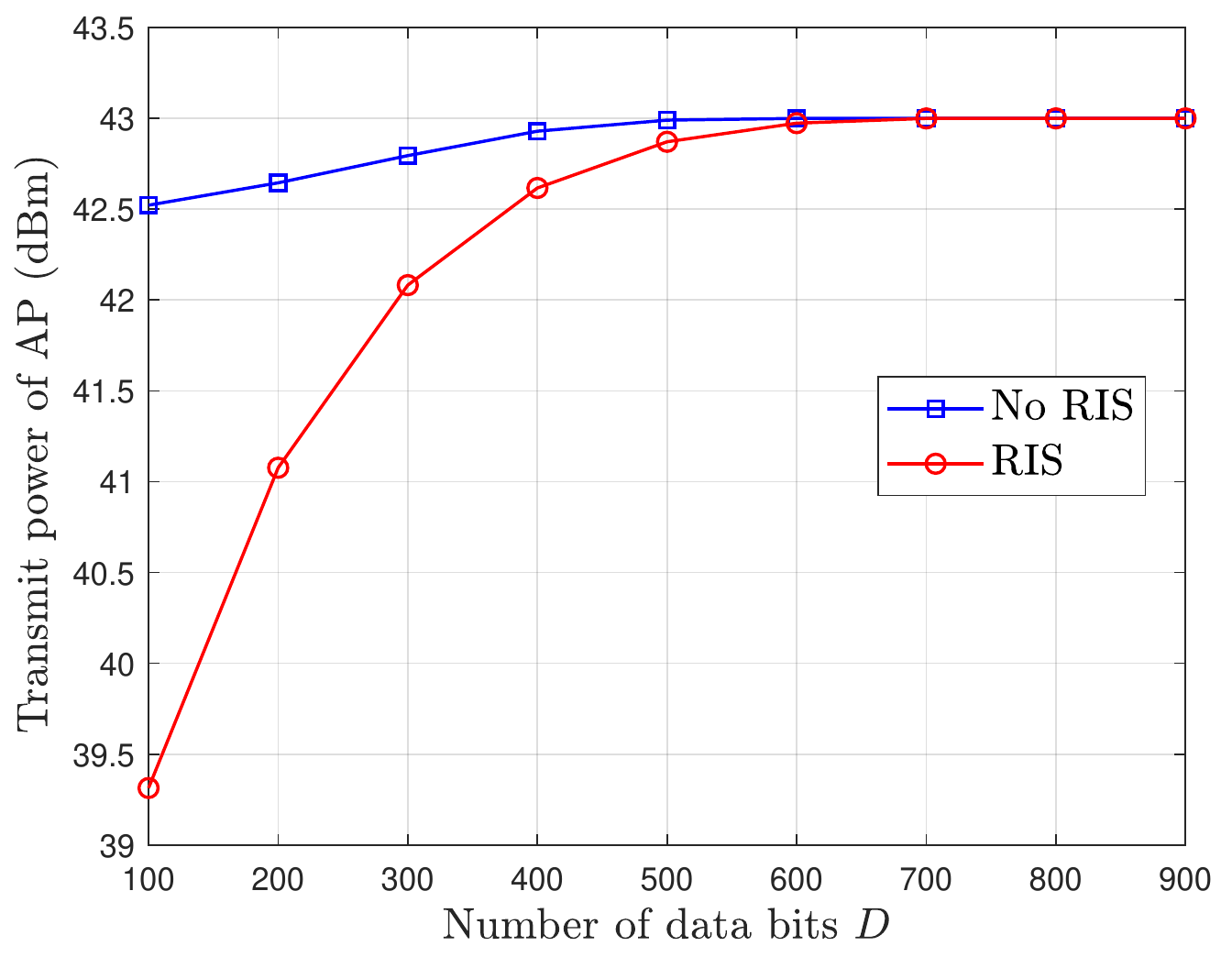}
	\caption{Average transmit power of AP under RIS and no RIS-assisted schemes.}
	\label{power_contrast}
\end{figure}
\section{Conclusion}
In this work, we proposed a two-stage RIS-aided D2D communication protocol to enable that all actuators in the smart factory can receive the critical control signals from AP successfully. Towards this end, the active beamforming at the AP and the passive phase shifts at the RIS needed to be jointly optimized. By tackling the indicator functions, we first obtained an alternative continuous form. Different from the common alternating optimization technique, we utilized the Frobenius inner product based equality to decouple the optimization variables. Then, we exploited the penalty-based approach for solving rank-one constraints. In addition, $\ell_1$-norm was applied to approximate $\ell_0$ norm and we added an extra term $\ell_1-\ell_2$ for sparsity. Numerical results validated the efficacy of the proposed protocol. Specifically, it can achieve reliable communication in a wider range of data bits as compared to the no RIS-assisted D2D communication protocol. Also, the proposed protocol can considerably reduce the transmit power consumption at the AP. The case of corresponding imperfect design will be studied in the future work.

{\footnotesize
	\bibliographystyle{IEEEtran}
	\bibliography{Reference}}

\end{document}